\documentclass[a4paper,USenglish,cleveref, autoref]{lipics-v2021}

\usepackage{amsmath}
\usepackage{amssymb,amsthm}
\usepackage{stmaryrd}
\usepackage[utf8]{inputenc}
\usepackage[T1]{fontenc}
\usepackage{mathtools}
\usepackage{microtype}
\usepackage{hyperref}
\hypersetup{colorlinks=true}

\usepackage{dsfont}
\usepackage{multicol}

\nolinenumbers

\author{Liron Cohen}{Ben-Gurion University,  Beer-Sheva, Israel}{cliron@bgu.ac.il}{0000-0002-6608-3000}{}

\author{Tomer Samara}{Ben Gurion University, Beer-Sheva, Israel}
{samarato@post.bgu.ac.il}{0009-0005-8644-6832}{}
\authorrunning{Cohen and Samara}

\funding{This work was supported by Grant No. 2876/25 from the Israeli Science Foundation (ISF).}

\newcommand{\ef}{\mathcal{E\!F}}

\newcommand{\Set}{\mathbf{Set}}
\newcommand{\Meas}{\mathbf{Meas}}
\newcommand{\must}{\mathsf{must}}
\newcommand{\Nat}{\mathbb{N}}

\newcommand{\defeq}{\mathrel{:=}}

\newcommand{\MC}{\mathsf{MC}}
\newcommand{\tr}{\mathsf{tr}}

\newcommand{\afterx}[4]{ {\langle \!\!\mathrel{\raisebox{-4pt}{$\diamond$}}}\, #1 \leftarrow #2  {\mathrel{\raisebox{-4pt}{$\diamond$}}\!\!\rangle}^{#4}\,   #3 }
\newcommand{\after}[3]{\afterx{#1}{#2}{#3}{}}
\newcommand{\afteras}[3]{\afterx{#1}{#2}{#3}{a.s}}
\newcommand{\aftermust}[3]{\afterx{#1}{#2}{#3}{must}}

\newcommand{\M}{\mathsf{M}}
\newcommand{\unit}{\eta}

\newcommand{\kleisli}{\mathbin{\gg\!=}}

\newcommand{\OmegaTV}{\Omega}
\newcommand{\diamondop}{\Diamond}

\newcommand{\maybe}[1]{#1_{\bot}}

\newcommand{\Code}{\mathsf{A}}     
\newcommand{\Ev}{\mathsf{E}}       

\newcommand{\entails}[3]{#1 \xrightarrow{#2} #3}
\newcommand{\xle}[1]{\xrightarrow{#1}}

\newcommand{\meas}{\mathrm{meas}}

\title{Evidence-Tracked Tape Semantics for Probabilistic Computation}

\begin{document}
\maketitle

\begin{abstract}
A standard intensional account of probabilistic computation represents a randomized program as a deterministic computation that consumes an explicit random tape. This yields a two-layer perspective: an intensional layer that makes reuse of randomness and correlation visible, and an extensional layer obtained by interpreting tapes under a chosen probability measure.
We develop an evidence-tracked tape semantics using the monadic-core-to-evidenced-frame pipeline (and its induced realizability tripos), obtaining a higher-order logic in which entailments are witnessed by uniform evidence transformers. Quantitative statements are recovered by interpretation: once a tape measure is fixed, probabilities and expectations arise by extracting numerical summaries from tape-indexed predicates, and entailments yield sound inequalities, with an almost-sure quotient supporting probability-one reasoning. We also study intensional principles that are lost at the level of laws, including proof-relevant transport along realizable tape-rewiring maps and a canonical splitting discipline for stream tapes enforcing independent draws. Finally, we relate tape-based reasoning to an extensional law semantics via pushforward, isolating a probability-one must abstraction as a sound summary of tape-based proofs.



\end{abstract}

\section{Introduction}

Randomized computation is now a standard part of programming practice, from cryptographic constructions and randomized algorithms to probabilistic programming, e.g.,~\cite{BonehShoup2020, Gordon2014, Saini2022, ProbSen2016}.
This has led to a large body of semantic and proof theoretic techniques for specifying and verifying probabilistic programs.
A recurring distinction in this landscape is between reasoning extensionally, by the output law (that is, output distribution) of a program, and reasoning intensionally, by how a program consumes randomness.

Most probabilistic program logics internalize quantitative specifications:
assertions speak about event probabilities, expectations, or real-valued preconditions, and proof rules are given in weakest-preexpectation or quantitative Hoare style (including relational variants), e.g.,~\cite{Barthe2013, Corin2006, Kaminski2016, Kozen1981, McIverMorgan05}.
These approaches are well suited for bounds and naturally align with law-based semantics.
A recent example is Probabilistic Outcome Logic~\cite{OutcomeLogic, zilberstein2026probabilistic}, which reasons about properties of outcome sets and supports algebraic reasoning for probabilistic and randomized nondeterministic constructs.

A different, equally classical viewpoint models randomness as an explicit input.
A probabilistic program is a deterministic function that consumes a random seed, a stream of bits, or a random tape~\cite{DeterminisitcStreamSample2023, Gianantonio24CCC, ProbSen2016}.
This viewpoint is convenient for modelling implementations, reuse of randomness, correlation, and deterministic transformations of the random source.
For example, drawing one random bit once and reusing it twice is observably different from drawing two independent bits, even if both programs may induce the same marginal distribution on a single output~\cite{Barthe2015}.
Such distinctions are invisible once one quotients a program by its law (i.e., the output distribution obtained by pushing forward a fixed tape measure through the program), but they are central to compilation and transformation arguments where the random source is reshaped.

Working with an explicit random tape yields a built-in separation of concerns: the operational semantics is measure-agnostic, and quantitative statements arise only after fixing a measure.
Much of the verification literature focuses on internal quantitative reasoning over laws, abstracting away intensional structure of randomness consumption~\cite{McIverMorgan05, ProbSen2016}.
Conversely, tape-based semantics is often used operationally, while the link from tape-indexed reasoning to probabilistic statements is handled outside the logic.
Our goal is a clean evidence-tracked account of tape-based computation with principled bridge results that make this link explicit.

Concretely, the paper develops an evidence-tracked higher-order logic whose entailments are uniform tape indexed refinements.
A proof of an entailment is witnessed by a single evidence transformer that works uniformly for every tape.
Tracking evidence makes entailment proof-relevant: an entailment is not just a relation between predicates, but is witnessed by a concrete code transformer mapping realizers of the premise to realizers of the conclusion.
This is essential here because key reasoning principles, such as tape rewiring and tape splitting, correspond to transformations of tape access. 
At the semantic level, these principles act by
reindexing computations and predicates, and at the evidenced-frame level, they require compatible
translations of evidence codes. In stronger concrete language instances, such translations may
further be internalized by compiler-like codes.
Probabilistic claims are obtained by an external extraction step, after fixing a probability measure on tapes.
This is a deliberate design choice: it keeps internal reasoning intensional and measure agnostic, while allowing the same tape level proof to be instantiated with different probabilistic readings.

Technically, we build on an established pipeline for obtaining models of higher-order intuitionistic logic via monadic cores, evidenced frames and realizability triposes~\cite{CohenGrunfeldKirstMiquey:FSCD25,Cohen2021}.
An evidenced frame~\cite{Cohen2021} packages a notion of computation together with a preorder of predicates and a notion of evidence that tracks entailment.
Applying a uniform construction to an evidenced frame yields a tripos, and hence a model of higher-order logic whose entailments carry explicit evidence.
In turn, the monadic combinatory algebra (MCA)  framework~\cite{CohenGrunfeldKirstMiquey:FSCD25} obtains  evidenced frames from monadic cores, built from a monadic combinatory algebra together with a choice of truth values and a suitable modality.
In this paper we take this pipeline as a tool.
Our novelty is not the generic construction, but its probabilistic instantiation for tape consuming computation and the resulting bridge theorems that connect uniform entailments to ordinary probability statements after interpretation.

In our tape setting, computations with outputs in $X$ are  measurable functions $m : R \to \maybe X$ over a measurable tape space $(R,\Sigma)$, where $X_{\bot} \defeq X \uplus \{\bot\}$ (i.e., $X$ with an added divergence value). 
Predicates are tape indexed truth values, that is, functions $R \to [0,1]$, ordered pointwise.
Entailment is therefore tape by tape: an entailment is witnessed by a single evidence transformer that works uniformly for every tape.
This makes the logic intensional and correlation aware: sequential composition threads the same tape through the continuation, so reuse of randomness is visible rather than erased.

To recover numeric probabilistic statements, we fix a probability measure $\rho$ on $(R,\Sigma)$ and apply extraction maps.
For instance, expectation turns a tape-indexed truth value $f : R \to [0,1]$ into a number $\mathbb E_\rho[f]$, and indicator predicates yield event probabilities.
We use the standard expectation interpretation to read quantitative consequences of our tape-level proof principles once a tape measure is fixed.
We also introduce ana almost-sure quotient on predicates that identifies functions agreeing outside a null set, yielding a probability-one fragment that avoids null-set artefacts while remaining compatible with our logical structure.

The tape level supports proof principles about randomness as a resource that are invisible at the level of laws.
We study two representative forms of intensional structure.
First, realizable tape maps let one rewrite the random source while transporting entailments with associated evidence translations, enabling proof reuse across changes of representation.
Second, for stream tapes we develop a canonical splitting construction that provides a disciplined notion of independent draws.
We show how splitting-based reasoning can be translated back into single-tape evidence.
These principles target transformation arguments, where preserving and reshaping the random source is the point.

In many settings one also wants to forget correlations and retain only an extensional law~\cite{giry81, ProbSen2016}.
We therefore study a distribution-level abstraction of the tape semantics.
The most robust fragment that aligns directly with tape reasoning without invoking expectation is the almost-sure, must-style fragment, capturing probability-$1$ properties such as almost-sure termination and support-based safety.
We make precise what information is retained and what is lost by passing from tapes to laws.

\paragraph*{Contributions.}
\begin{itemize}
\item We develop an evidence-tracked, tape-based instantiation of the monadic-core-to-evidenced-frame pipeline, with tape-consuming computations and tape-indexed truth values, where entailments are uniform tape-by-tape refinements witnessed by evidence transformers.

\item We introduce realizable tape maps and show that they transport evidenced entailments together with  corresponding translations on evidence, enabling proof reuse across changes of random-source representation.

\item For bit-stream tapes we study a canonical splitting construction that supports disciplined independent draws, and we show how splitting-based arguments  can be transported back to the single-tape setting.

\item After fixing a probability measure on tapes, we quotient tape truth values by almost-sure equality, yielding a probability-one interface that avoids null-set artefacts and supports robust reasoning up to measure-zero differences.

\item With a tape measure in place, we use expectation and event-probability extraction to interpret tape-indexed truth values numerically, and we show that tape-level entailments yield the corresponding quantitative inequalities under this interpretation.
\end{itemize}

\textbf{Outline.}
Section~\ref{sec:background} recalls the monadic core to evidenced frame to tripos construction used in the paper.
Section~\ref{sec:tape} develops the tape model, extraction principles, and the almost sure quotient.
Section~\ref{sec:tape:streams} treats streams, splitting, and transport along realizable tape maps. 
Section~\ref{sec:dist-final} relates tape computations to laws and formulates the probability-one must abstraction. 
Then, Section~\ref{sec:conc} concludes and discusses future work.

\section{Background}\label{sec:background}

Our goal is to obtain semantic models of higher-order intuitionistic logic in which specifications relate to probabilistic computations.
We do so via an established two-step route from effectful computation to higher-order logic:
(i) from a \emph{monadic core} (an effectful combinatory structure plus a truth-value modality~\cite{CohenGrunfeldKirstMiquey:FSCD25}) one builds an \emph{evidenced frame} (a minimal proof-relevant structure abstracting the entailment rules needed for realizability~\cite{Cohen2021}),
and (ii) from an evidenced frame one obtains a \emph{realizability tripos}, a standard categorical presentation of higher-order intuitionistic logic~\cite{HylandJohnstonePitts80,PittsTriposRetrospect2002}.
In this paper we use evidenced frames as a lightweight interface: they expose \emph{proof-relevant} entailments while staying close to program semantics.
The tripos construction mainly serves as semantic motivation, explaining why evidenced-frame entailments are implications in a higher-order logic and why evidence terms matter.
This section reviews the key components from the 
 monadic-core-to-evidenced-frame pipeline.

\subsection{Monadic cores}
\label{sec:monadic-cores}
Let $M$ be a monad on $\Set$~\cite{Moggi1991}. 
We write $\unit_\Code : \Code \to M(\Code)$ for the unit and $(-)\kleisli(-) : M(\Code) \to (\Code \to M(B)) \to M(B)$ for Kleisli composition (bind), often written $m \kleisli f$ for $m \in M(\Code)$ and $f : \Code \to M(B)$. 
An \emph{$M$-combinatory algebra} (MCA) consists of a set of codes $\Code$ equipped with a monadic application operation
$\cdot : \Code \times \Code \to M(\Code_\bot)$ satisfying a combinatory completeness requirement that ensures enough expressivity to interpret the untyped $\lambda$-calculus. 
(Equivalently, this can formulated via the existence of combinators witnessing the usual $S/K$ structure internalized through $M$ (see~\cite[Prop.~7]{CohenGrunfeldKirstMiquey:FSCD25}).

To turn an MCA into a proof-relevant calculus one has to provide additional ingredients:
\begin{itemize}
\item a truth-value object $(\Omega,\le)$ (a complete Heyting prealgebra), and
\item an \emph{$M$-modality} $\diamondop$-valued in $\Omega$ that internalizes effectful sequencing at the level of truth values.\footnote{We use the  terminology of~\cite{pitts1991evaluation,CohenGrunfeldKirstMiquey:FSCD25}.
A monad $M$ together with an $M$-modality induces a monad over the appropriate fibration~\cite{aguirre2020weakest}, so an $M$-modality is also a modality in the ``monad as a modality'' sense.}
\end{itemize}

\begin{definition}[$M$-modality]
\label{def:m-modality}
An $M$-modality over $\Omega$ is a natural transformation in $X$,
$$
\diamondop_X : M(X) \to (X \to \Omega) \to \Omega,
\qquad \text{and we note }\quad \after{x}{m}{\varphi\left(x\right)}\defeq \diamondop_X\left(m\right)\left(\varphi\right)
$$
satisfying, for all sets $A,B$, all predicates \(\varphi, \varphi_1, \varphi_2:A\to\OmegaTV\),
all Kleisli maps \(f:A\to \M(B)\), all \(a\in A\), and all \(m\in \M(A)\):
\begin{enumerate}
  \item \textbf{After-Return.}
  \(\varphi(a)\ \le\ \after{x}{\mu(a)}{\varphi(x)}\).
  \item \textbf{After-Bind.}
  $\after{x}{m}{(\after{y}{f(x)}{\varphi(y)})} \leq \after{y}{(m \kleisli f)}{\varphi(y)}$
  \item \textbf{Internal monotonicity.}
  $$
    \Big(\bigwedge_{x\in A} (\varphi_1(x)\Rightarrow \varphi_2(x))\Big)
    \ \le\
    \Big(\after{x}{m}{\varphi_1(x)}\Rightarrow \after{x}{m}{\varphi_2(x)}\Big)
  $$
\end{enumerate}
\end{definition}
Intuitively, $\after{x}{m}{\varphi\left(x\right)}$ should be read as 
`run the computation $m$, obtain a value $x$, and then evaluate the post-condition $\varphi(x)$'.
Thus $\diamondop$ acts as a weakest-precondition style predicate transformer for effectful computations.
The axioms express sound reasoning principles for such sequencing:
\emph{After-Return} states that pure values satisfy their post-conditions,
\emph{After-Bind} ensures compatibility with sequential composition (Kleisli bind),
and \emph{Internal monotonicity} guarantees that strengthening post-conditions preserves validity.

Finally, effectful realizability typically restricts which codes are allowed to serve as \emph{evidence} for entailments through the notion of a separator.
In the context of this paper we take the separator to be the full set $C$.
A \emph{monadic core} is then a tuple 
$
\MC \;=\; (\Code,\Omega,\diamondop).
$

\subsection{From monadic cores to evidenced frames}
\label{sec:mc-to-ef}

Evidenced frames refine Heyting style semantics by recording \emph{proof relevant} entailments~\cite{Cohen2021}.
Intuitively, an entailment $\entails{\varphi}{e}{\psi}$ is witnessed by a concrete evidence term $e$ that uniformly converts any realizer of $\varphi$ into a realizer of $\psi$.
More generally, an evidenced frame can be regarded as a frame, or (more accurately) a complete Heyting algebra, in which one provides evidence that an element is smaller than another, and there can be multiple forms of evidence for the same fact.
This proof relevance is what supports realizability style program extraction and enables comparing different effectful semantics at the level of uniform evidence.

Concretely, an evidenced frame is a tuple $(\Phi,\Ev,\cdot \xle{\cdot} \cdot)$ where $\Phi$ is a set of \emph{propositions}, $E$ is a set of \emph{evidence terms},
and $\phi_1 \xle{e} \phi_2$ is a ternary \emph{evidence relation} on
$\Phi\times E\times \Phi$, together with a designated
set of proposition-formers and evidence constructors.
Given a monadic core $\MC=(\Code,\Omega,\diamondop)$ one obtains an evidenced frame
$\ef_\MC=(\Phi_\MC,\Ev_\MC,\entails{\cdot}{\cdot}{\cdot})$ as follows:
propositions are $\Omega$-valued predicates on codes (i.e.~$\Phi_\MC \defeq \Code \to \Omega$),
evidence are codes (i.e.~$\Ev_\MC \defeq \Code$),
and entailment is the tracked refinement condition
$$
\varphi \xle{e} \psi
\quad\!\!\Longleftrightarrow\quad
\forall c\in\Code.\ 
\varphi(c) \le \after{x}{e \cdot c}{\psi(a)}.
$$
Intuitively, $e$ is a uniform evidence transformer: given any realizer $c$ satisfying $\varphi$,
the computation $(e\cdot c)$ produces outputs that satisfy $\psi$ under the modality $\diamondop$.

\subsection{From evidenced frames to realizability triposes}
\label{sec:ef-to-tripos}

Every evidenced frame $\ef=(\Phi,E,\entails{\cdot}{\cdot}{\cdot})$ induces a tripos $P_\ef:\Set^{op}\to \mathbf{HeytPreord}$ by taking, for each set $X$,
the fiber $P_\ef(X)$ to consist of predicates $X\to \Phi$ ordered by \emph{pointwise evidenced entailment}:
$
\alpha \le \beta
\quad\text{iff}\quad
\exists e\in E.\ \forall x\in X.\ \entails{\alpha(x)}{e}{\beta(x)}
$.
Reindexing is by precomposition, and the generic tripos structure (Heyting operations, quantifiers) is defined uniformly from the evidenced-frame structure~\cite{Cohen2021}.

A closed proposition in the induced logic corresponds to an element of $\Phi$ (take $X=1$).
Thus, if $\entails{\varphi}{e}{\psi}$ holds in the evidenced frame, then the induced tripos validates the implication $\varphi \Rightarrow \psi$,
and the same code $e$ is precisely the \emph{realizer} witnessing this implication.
More generally, a single evidence term $e$ establishing $\alpha(x)\Rightarrow \beta(x)$ uniformly for all $x\in X$ is exactly what makes the induced higher-order reading nontrivial. 
Thus, 
it is uniform proof-relevant structure, not merely pointwise semantic validity.

The tripos is the model of higher-order logic and  the evidenced frame is the proof-relevant layer that remembers \emph{how} entailments are witnessed.
This is what lets us state and prove compilation and transport results (one evidence term that works uniformly across all realizers),
and later read them as sound principles in the induced logic.

\subsection{Measure-theoretic conventions}\label{sec:meas-theory}

This section recalls the basic probabilistic notions used in the paper.
We keep the presentation lightweight and oriented toward discrete and tape-based semantics.
A $\sigma$-algebra on a set $X$ is a nonempty collection $\Sigma_X \subseteq \mathcal{P}(X)$
that is closed under complements and countable unions.
A \emph{measurable space} is a pair $(X,\Sigma_X)$ where $X$ is a set and $\Sigma_X$ is a $\sigma$-algebra on $X$.
For $(X,\Sigma_X)$ and $(Y,\Sigma_Y)$  measurable spaces, 
a function $f : X \to Y$ is \emph{measurable} if for every $A \in \Sigma_Y$, $f^{-1}(A) \in \Sigma_X$.
A \emph{measure} on $(X,\Sigma_X)$ is a function $\mu:\Sigma_X \to [0,\infty]$ such that $\mu(\emptyset)=0$,
and for every countable family of pairwise disjoint sets $(E_i)_{i\in\mathbb{N}}$,
$
\mu\Big(\bigsqcup_{i} E_i\Big)=\sum_i \mu(E_i).
$
A \emph{probability measure} is a measure $\mu$ with $\mu(X)=1$, and a \emph{probability space} is a triple $(X,\Sigma_X,\mu)$.
In the discrete setting used throughout this paper, a probability distribution on a countable set $A$
is a function $\mu : A \to [0,1]$ such that $\sum_{a\in A} \mu(a)=1$.
This presentation is sufficient for all tape-based constructions in this work.

Let $f : X \to Y$ be a measurable function and $\mu$ a probability measure on $X$.
The \emph{pushforward measure} (or \emph{law}) of $f$ is the measure $f_\ast \mu$ on $Y$ defined by $f_\ast \mu(B) \defeq \mu(f^{-1}(B))$.
The law captures the observable output distribution of a computation:
it forgets how randomness is consumed internally and retains only the distribution of results.
Crisp predicates correspond to characteristic functions.
For a subset $P \subseteq X$, its indicator is the function
\[
\mathbf{1}_P(x) =
\begin{cases}
1 & x \in P,\\
0 & x \notin P.
\end{cases}
\]

\section{Random tapes via a reader monad}
\label{sec:tape}

This section develops our main probabilistic model based on \emph{random tapes}.
The central idea is to represent probabilistic computation \emph{extensionally} as deterministic computation that reads from
an explicit source of randomness.
Concretely, a probabilistic program is modeled as a (possibly partial) function that consumes a random tape and
produces an observable output.
This is the familiar ``random variable'' view of probabilistic computation, and it is
standard in denotational semantics (see, e.g.,~\cite{ Gianantonio24CCC, ThomasEhrhard2020, Gill1974, ProbSen2016}).
The advantage for our purposes is conceptual and proof-theoretic: randomness is made explicit and can be manipulated
by program transformations (such as tape reindexing and tape splitting), while the underlying program dynamics remains deterministic.

\subsection{The random tape evidenced frame}
\label{subsec:random-tape-ef}

Fix a set $R$ of random tapes equipped with a $\sigma$-algebra $\Sigma$.
A random tape is an abstract representation of a complete supply of random choices for a run of a program.
Formally, we model tapes by a set $R$ whose elements are the possible realizations of randomness.
The canonical example is the space of infinite bit streams $R\defeq 2^{\mathbb N}$, where the program reads bits from the stream as needed.
More generally, one may take $R$ to be any space of infinite streams over a finite alphabet, or a product space encoding multiple sources of randomness.
When extracting numeric probabilities, we additionally fix a probability measure
$\rho$ on $(R,\Sigma)$ specifying how tapes are sampled (e.g.\ the fair-coin product measure on $2^{\mathbb N}$).
Operationally, a deterministic evaluator consumes a tape $r\in R$ by reading a finite prefix (or more, as the computation proceeds),
so sampling $r\sim\rho$ induces the usual probabilistic behavior.

For a countable $X$, $X_{\bot} \defeq X \uplus \{\bot\}$, i.e., $X$ with an added divergence value. We use this to `totalize' the output set, and we equip $X_{\bot}$ with the discrete $\sigma$-algebra.
A tape computation with outputs in $X$ is a measurable function $m:R\to X_{\bot}$.
Define the $\Set$-monad $\M$ by
$$
  \M(X) \defeq \{\, m:R\to X_{\bot}\mid m \text{ is } \Sigma\text{-measurable}\,\}.
$$
We emphasize that $\M$ depends only on the measurable structure $(R,\Sigma)$ (not on any choice of measure on $R$).
We sometimes write $\M_R$ to make this dependence explicit.

This is the (partial) \emph{reader monad} on $R$ (restricted to measurable functions), familiar from semantics and functional
programming: a computation is a function that reads from a shared environment, here the environment is the random tape~\cite{Kozen1981, DeterminisitcStreamSample2023}.
The crucial point is that randomness is not internalized as probabilistic choice in the monad.
Instead, the monad exposes randomness as an explicit parameter $r\in R$.
Reusing the same tape $r$ across sequential composition makes correlations explicit, and operations such as tape reindexing
and splitting can be expressed as deterministic transformations on tapes.

The unit $\unit_X:X\to \M(X)$ is the constant function
 $ \unit_X(x)(r)\defeq x$.
Given $f:X\to \M(Y)$ and $m\in \M(X)$, the Kleisli bind threads the \emph{same} tape through sequential composition:
$$
  (m\kleisli f)(r)\defeq
  \begin{cases}
    \bot & \text{if } m(r)=\bot,\\
    f(x)(r) & \text{if } m(r)=x\in X.
  \end{cases}
$$

\begin{lemma}\label{lem:reader-monad}
$\M$ forms a monad.
\end{lemma}

   Fix a set of codes $\Code$  with an MCA application over $\M$,
$
  \cdot : \Code \times \Code \to \M(\Code) \subseteq (R\to \Code_{\bot})
$.
Thus for $c_1,c_2\in \Code$,  $c_1\cdot c_2$ denotes a measurable tape computation
$r\mapsto (c_1\cdot c_2)(r)$ in $\Code_{\bot}$.
From this, one can construct a complete Heyting prealgebra as follow.
Take $[0,1]$ with the usual order $\le$ and G\"odel implication ($I_G(a, b) = 1 \text{ if } a \leq b $ and $b$ otherwise).
We take truth values to be tape-indexed quantitative values,
$
  \OmegaTV_R \defeq  (R \to [0,1])
$, 
ordered pointwise.
Intuitively, a truth value $\alpha\in\OmegaTV_R$ is a \emph{random variable of success levels}:
given a tape $r$, the number $\alpha(r)$ is the success level achieved on that particular random tape.
Pointwise order then corresponds to tape-by-tape improvement of success.
Because meets, joins, and implication are computed pointwise, $\OmegaTV_R$ forms a complete Heyting prealgebra.

\begin{definition}
Define an operation $\diamondop$ valued in $\OmegaTV_R$ as follows.
For every countable set $X$,  $m\in \M(X)$, and  predicate $\psi:X\to \OmegaTV_R$, define
$\after{x}{m}{\psi(x)}\in \OmegaTV_R$ by: for each tape $r\in R$,
$$
  \Big(\after{x}{m}{\psi(x)}\Big)(r)
  \defeq
  \begin{cases}
    0 & \text{if } m(r)=\bot,\\
    \psi(x)(r) & \text{if } m(r)=x\in X.
  \end{cases}
$$
\end{definition}

That is, on tape $r$ the modality $\diamondop$ either diverges (yielding success level $0$), or returns a value $x$ and then evaluates the post-condition at the same tape $r$. This is the standard partial-correctness (weakest-precondition) modality for partial computations, defined here pointwise in the tape (e.g.,~\cite{Kozen1981, Zilberstein2023}).

\begin{lemma}\label{lem:diamop-modality}
The operator $\diamondop$ is an $\M$-modality.
\end{lemma}

\begin{definition}[The random tape evidenced frame]
\label{def:ef-tape}
Let $(\Code,\cdot)$ be an MCA over $\M$ and let $\diamondop$ be the above modality valued in $\OmegaTV_R$.
The induced evidenced frame is denoted $\ef_{\mathsf{tape}}$.
\end{definition}
The propositions in $\ef_{\mathsf{tape}}$ are tape-indexed quantitative predicates
$$
  \Phi_{\mathsf{tape}} \defeq \Code \to \OmegaTV_R = \Code \to (R\to[0,1]).
$$
For $\varphi\in\Phi_{\mathsf{tape}}$ and $c\in \Code$, the value $\varphi(c)(r)$ is the success level asserted for running $c$
on tape $r$.
A postcondition is a predicate $\psi:\Code \to(R\to[0,1])$.
For instance, if $P\subseteq \Code$ is a crisp property of outputs (i.e., that takes only Boolean truth values), we may take the tape-invariant postcondition
$\psi_P(a)(r)\defeq \mathbf{1}_P(a)$ 
Then for $c_1,c_2\in \Code$, the truth value
$$
    \after{a}{(c_1 \cdot c_2)}{\psi_P(a)} \in (R\to[0,1])
$$
maps each tape $r$ to $1$ exactly when $(c_1\cdot c_2)(r)\in P$, and to $0$ otherwise (including divergence).
An entailment
$
  \entails{\varphi}{e}{\psi}
$
 unfolds to the pointwise condition
$$
  \forall c\in A,\ \forall r\in R.\quad
  \varphi(c)(r)\ \le\
  \Big(\after{a}{(e \cdot c)}{\psi(a)} \Big)(r).
$$
Thus entailment is a uniform tracked refinement principle parameterized by the tape: the same evidence $e$ works for all inputs
$c$, and the refinement holds tape-by-tape.
For example, fix a crisp postcondition  $\psi_P$ as above and take $\varphi$ to be the constant predicate
$\top(c)(r)\defeq 1$.
Then $\entails{\top}{e}{\psi_P}$ asserts that for every input code $c$ and every tape $r$, if $e\cdot c$ terminates on $r$
then its output lies in $P$ (and divergence is treated as failure with score $0$).

Entailments in $\ef_{\mathsf{tape}}$ are tape-indexed refinements, the next subsection explains how to extract numeric probability bounds from them once a measure $\rho$ on $(R,\Sigma)$ is fixed.

\subsection{From tape-indexed truth values to numeric probabilities}
\label{sec:bridge:extraction}

In this section we fix a probability measure $\rho$ on $(R,\Sigma)$.
The tape evidenced frame is built over $\Omega_R\defeq (R\to[0,1])$ to keep completeness and avoid measurability side conditions.
Measurability is needed only for probabilistic interpretation.
Accordingly, let $\Omega_{\meas}\subseteq \Omega_R$ be the measurable functions.
All expressions to which we apply expectation or almost-sure reasoning lie in $\Omega_{\meas}$ since 
$\Omega_{\meas}$ is closed under pullback along measurable tape maps and under the tape modality.

\begin{lemma}
\label{lem:diamond-meas}
Let $X$ be countable and let $m:R\to X_\bot$ be measurable. If $\psi:X\to \Omega_{\meas}$
then $\langle \diamond x \leftarrow m\rangle \psi(x)\in \Omega_{\meas}$.
\end{lemma}

To recover ordinary numeric probability statements from random-tape evidenced frame entailments, we aggregate tape-indexed truth values by expectation with respect to $\rho$.
This is a standard monotone interpretation that yields sound numeric probability bounds, which we record because it is the bridge we use later to read probabilistic consequences of evidence transport and splitting.
For $\alpha\in \Omega^{\meas}$ define its expectation by
$$
  E_\rho(\alpha) \defeq \int_R \alpha(r)\, d\rho(r)\ \in [0,1].
$$
For a tape-indexed predicate $\varphi:\Code\to\Omega^{\meas}$ define its extracted numeric predicate
$E_\rho(\varphi):\Code\to[0,1]$ by
 $ (E_\rho(\varphi))(c) \defeq E_\rho(\varphi(c))$.
Note that expectation is monotonic, i.e., if $\alpha,\beta\in\Omega^{\meas}$ and $\alpha(r)\le \beta(r)$ for all $r\in R$, then
$E_\rho(\alpha)\le E_\rho(\beta)$.

Since expectation is monotone, entailments between measurable predicates yield inequalities after extraction:
if $\varphi,\psi:A\to\Omega_{\meas}$ and $\ef_{\mathsf{tape}}\vdash \entails{\varphi}{e}{\psi}$, then for every $c\in A$,
$$
  E_\rho(\varphi(c)) \ \le\ E_\rho\!\Big(\after{a}{(e\cdot c)}{\psi(a)}\Big),
$$
(where the right-hand side is well defined by Lemma~\ref{lem:diamond-meas}).
We will use this repeatedly to turn the tape-level theorems (transport, splitting) into numeric bounds.

But while extracting probability bounds is sound, it is not complete. That is, numeric inequalities after extraction do not, in general, determine the underlying
pointwise order  and therefore do not entail an evidenced-frame entailment.
This is because an inequality may hold after applying $E_\rho$
even if it fails on a set of tapes of small (but nonzero) $\rho$-measure.

For $c', c\in \Code$, since $(c' \cdot c)\in \M(\Code)$ is a measurable function $R\to \Code{\bot}$ (with $\Code_{\bot}$ discrete), it induces a probability
measure on $\Code_{\bot}$ by pushforward: for $B\subseteq \Code_{\bot}$, 
$
  \mu_{c',c}(B) \defeq \rho\big((c' \cdot c)^{-1}(B)\big)
$.
In particular, for each $a\in \Code_{\bot}$,
 $ \mu_{c',c}(\{a\}) \;=\; \rho\big(\{\,r\in R \mid (c' \cdot c)(r)=a\,\}\big)$.
Restricting $\mu_{c',c}$ to $\Code$ yields a subdistribution (the missing mass is the divergence probability $\mu_{c',c}(\{\bot\})$).

As standard, $(e\cdot c)$ induces a subdistribution $\mu_{e,c}$ by pushforward.
For any numeric postcondition $\psi_0:A\to[0,1]$, the extracted value
$
E_\rho\!\Big(\,\langle \after{a}{(e \cdot c)}{\psi_0^\sharp(a} \Big)
$
coincides with the integral $\int \psi_0\, d\mu_{e,c}$ (and with the corresponding sum when $A$ is countable).
For a crisp postcondition $\psi_0=\mathbf{1}_P$, this is simply the event probability $\mu_{e,c}(P)$.

\subsection{Almost-sure extensionality for tape truth values}
\label{sec:tape:as-quotient}

To connect this intensional layer to probabilistic semantics, we move to a second, extensional layer.
Once a measure $\rho$ on $(R,\Sigma)$ is fixed, changing a computation or a truth value on a $\rho$-null set should not change its probabilistic meaning.
We therefore quotient tape truth values by almost-sure equality, obtaining a Heyting structure $\Omega_{a.s.}$ and an \emph{almost-sure} evidenced frame.
From this extensional layer there are two systematic ways to read probabilistic content:
a numeric extraction map (expectation and event probabilities), and a Boolean \emph{probability-one collapse} that records whether a truth value holds almost surely.
We use the probability-one collapse as the interface to the later must-style, law-level abstraction.
\begin{definition}
Let $ \alpha, \beta \in \Omega_{\meas}$.
Define $\alpha \equiv_{\mathrm{a.s.}} \beta$ iff
$
\rho(\{\, r\in R \mid \alpha(r)\neq \beta(r)\,\}) = 0
$.
Write $[\alpha]$ for the $\equiv_{\mathrm{a.s.}}$-equivalence class of $\alpha$, and let
$
\Omega_{a.s.}\defeq \Omega_{\meas}/\equiv_{a.s.}
$.
Define an order on $\Omega_{\mathrm{a.s.}}$ by
$$
[\alpha] \le_{\mathrm{a.s.}} [\beta]
\quad\Longleftrightarrow\quad
\rho(\{\, r\in R \mid \alpha(r)\le \beta(r)\,\}) = 1.
$$
\end{definition}
Intuitively, $[\alpha]$ records the truth value $\alpha$ \emph{up to null sets of tapes}, and
$[\alpha]\le_{\mathrm{a.s.}}[\beta]$ means that $\alpha(r)\le\beta(r)$ holds for $\rho$-almost every tape $r$.
There is a canonical collapse from $\Omega_{a.s.}$ to Boolean truth values, sending a class $[\alpha]$ to $1$ exactly when $\alpha(r)=1$ for $\rho$-almost every tape $r$.
Intuitively, this forgets quantitative information and keeps only the probability-one content.
This collapse is the bridge point between tape-based reasoning and the must-style, distribution-level abstraction in Section~\ref{sec:dist-final}.

\begin{definition}
Let $\mathcal{A}\subseteq  \Omega_{\meas}
$ be any family of measurable functions.
An \emph{essential upper bound} of $\mathcal{A}$ is a function $u\in\Omega$ such that
$\rho(\{r\mid \alpha(r)\le u(r)\})=1$ for all $\alpha\in\mathcal{A}$.
An \emph{essential supremum} of $\mathcal{A}$ is an essential upper bound $u$ such that for every essential upper bound
$v$ of $\mathcal{A}$ we have $u\le_{\mathrm{a.s.}} v$.
We write $u=\mathrm{ess\,sup}\,\mathcal{A}$.
Essential infimum $\mathrm{ess\,inf}\,\mathcal{A}$ is defined dually.
\end{definition}
Essential suprema/infima exist in $\Omega$ for arbitrary families of $[0,1]$-valued measurable functions,
and are unique up to $\equiv_{\mathrm{a.s.}}$.
(Equivalently: $\Omega_{\mathrm{a.s.}}$ is a complete lattice, with arbitrary joins/meets represented by
$\mathrm{ess\,sup}$ and $\mathrm{ess\,inf}$.)

We next show that $\Omega_{\mathrm{a.s.}}$ is a complete Heyting prealgebra.

\begin{proposition}
\label{prop:omega-as-heyting}
$\Omega_{\mathrm{a.s.}}$ is a complete Heyting prealgebra.
For $[\alpha],[\beta]\in \Omega_{\mathrm{a.s.}}$, define: 
$$
[\alpha]\wedge[\beta] \defeq \big[\,\min(\alpha,\beta)\,\big],
\qquad
[\alpha]\vee[\beta] \defeq \big[\,\max(\alpha,\beta)\,\big],
\qquad
[\alpha]\Rightarrow[\beta] \defeq \big[\,\alpha \Rightarrow \beta\,\big],
$$
where $\min,\max$ are pointwise, and $\Rightarrow$ is the pointwise G\"odel implication on $[0,1]$.
For any family $\{[\alpha_i]\}_{i\in I}\subseteq \Omega_{\mathrm{a.s.}}$, define: 
$
\bigvee_{i\in I}[\alpha_i] \defeq \big[\,\mathrm{ess\,sup}_{i\in I}\alpha_i\,\big],
$ and $
\bigwedge_{i\in I}[\alpha_i] \defeq \big[\,\mathrm{ess\,inf}_{i\in I}\alpha_i\,\big]
$.
\end{proposition}
\begin{proof}
Well-definedness: if $\alpha\equiv_{\mathrm{a.s.}}\alpha'$ and $\beta\equiv_{\mathrm{a.s.}}\beta'$ then
$\min(\alpha,\beta)\equiv_{\mathrm{a.s.}}\min(\alpha',\beta')$, similarly for $\max$ and for G\"odel implication,
because these operators are defined pointwise and are continuous on $[0,1]$.
Completeness: for any family $\{\alpha_i\}$, essential sup/inf exist and are unique up to $\equiv_{\mathrm{a.s.}}$,
hence define joins/meets in the quotient.
Heyting adjunction: for $\alpha,\beta,\gamma\in\Omega$, the pointwise G\"odel implication satisfies
$\min(\alpha(r),\gamma(r))\le \beta(r)$ iff $\gamma(r)\le(\alpha(r)\Rightarrow\beta(r))$ for each $r$,
therefore $\alpha\wedge\gamma\le_{\mathrm{a.s.}}\beta$ iff $\gamma\le_{\mathrm{a.s.}}(\alpha\Rightarrow\beta)$.
\end{proof}

\begin{lemma}
\label{lem:modality-respects-as}
Let $X$ be countable and let $m:R\to X_\bot$ be measurable.
If $\psi,\psi':X\to \Omega_{\meas}$ satisfy $\psi(x)\equiv_{\mathrm{a.s.}}\psi'(x)$ for all $x\in X$, then
$
\after{x}{m}{\psi(x)}
\ \equiv_{\mathrm{a.s.}}\
\after{x}{m}{\psi'(x)}.
$
\end{lemma}
For a postcondition $\bar\psi:X\to \Omega_{\mathrm{a.s.}}$, choose measurable representatives
$\psi:X\to \Omega_{\meas}$ with $[\psi(x)]=\bar\psi(x)$ for all $x$.
By Lemma~\ref{lem:modality-respects-as} we can define
$$
\afteras{x}{m}{\bar{\psi}(x)}
\ \defeq\
\Big[\after{x}{m}{\psi(x)}]\ \in\ \Omega_{\mathrm{a.s.}}.
$$

\begin{lemma}
\label{prop:quotient-modality-axioms}
The operation $\diamond^{\mathrm{a.s.}}$ is an $M$-modality.
\end{lemma}
Therefore, via the monadic-core-to-evidenced-frame construction using $\Omega_{\mathrm{a.s.}}$ and $\diamond^{\mathrm{a.s.}}$, we obtain an almost-sure tape evidenced frame $\ef_R^{\mathrm{a.s.}}$.
Concretely, define the extensional propositions to be
$
\Phi_{\mathrm{a.s.}} \defeq \Code \to \Omega_{\mathrm{a.s.}}.
$
Entailment in the extensional tape evidenced frame is:
$$
\entails{\varphi}{e}{\psi}\ \text{ holds iff }\ \forall c\in\Code.\ \varphi(c)\le_{\mathrm{a.s.}}
 \after{a}{(e \cdot c)}{\psi(a)}.
$$
Equivalently, for each fixed input $c$ the required inequality holds for $\rho$-almost every tape $r$.
This quotient therefore removes null-set artefacts while preserving the full higher-order (Heyting) logical structure
needed by the evidenced-frame-to-tripos pipeline.

\begin{proposition}
\label{prop:pointwise-to-as}
$$
\ef_R \vdash \entails{\varphi}{e}{\psi}
\quad\Longrightarrow\quad
\ef_R^{\mathrm{a.s.}} \vdash \entails{[\varphi]}{e}{[\psi]}.
$$
\end{proposition}
\begin{proof}
Unfolding $\ef_R \vdash \entails{\varphi}{e}{\psi}$ we get that for every  $c\in \Code$ and  $r \in R$,  $\varphi(c)(r) \leq \Big(\after{a}{(e \cdot c)}{\psi(a)}\Big)(r)$. 
To show $\ef_R^{\mathrm{a.s.}} \vdash \entails{[\varphi]}{e}{[\psi]}$, we need to show that for every $c \in \Code$, $\varphi(c) \leq_{\mathsf{a.s}} \after{a}{(e \cdot c)}{\psi(a)}$. By definition of $\leq_{\mathsf{a.s}}$, it suffices to show $\rho(\{ r \in R \mid \varphi(c)(r) \leq \Big(\after{a}{(e \cdot c)}{\psi(a)}\Big)(r)\}) = 1$, which holds because $\{ r \in R \mid  \varphi(c)(r) \leq \Big(\after{a}{(e \cdot c)}{\psi(a)}\Big)(r)\} = R$.
\end{proof}

The implication of Proposition~\ref{prop:pointwise-to-as} is strict in general.
Let $R=2^{\mathbb N}$ with the fair product measure $\mu$, and let $r_0\defeq 000\cdots$.
Define measurable truth values $\alpha',\beta':R\to[0,1]$ by
$$
\alpha'(r)\defeq 1,
\qquad
\beta'(r)\defeq
\begin{cases}
0 & \text{if } r=r_0,\\
1 & \text{if } r\neq r_0.
\end{cases}
$$
Then $\alpha'\le_{\mathrm{a.s.}}\beta'$ because $\mu(\{r_0\})=0$ and $\alpha'(r)\le\beta'(r)$ holds for $\mu$-almost every $r$,
but $\alpha'\not\le\beta'$ pointwise since $\alpha'(r_0)=1>\beta'(r_0)=0$.
Thus the almost-sure quotient removes null-set artefacts that are invisible to probabilistic semantics but would block
entailments in the intensional, pointwise tape logic.

Since expectation identifies almost-surely equal random variables, the extraction map
$$
E_\rho:\Omega_{\mathrm{a.s.}}\to[0,1],
\qquad
E_\rho([\alpha]) \defeq \int_R \alpha(r)\,d\rho(r),
$$
is well-defined and monotone.
Thus numeric probability bounds can be extracted from entailments in the extensional tape evidenced frame without
reintroducing pointwise (tape-by-tape) distinctions.

\subsection{Tape maps transport evidenced entailments}\label{sec:tape:tmap}

In the tape evidenced frame, the probabilistic behaviour of a code is mediated entirely by the
tape space: a code does not \emph{sample} randomness internally, but instead deterministically
\emph{consumes} a tape.
It is therefore natural to compare two tape models by a deterministic map that rewires one
tape space into another.
We show that a realizable tape map induces a computable translation of entailment evidence.

Let $(R,\Sigma)$ and $(R',\Sigma')$ be tape spaces.
Given a measurable function $\kappa:(R,\Sigma)\to (R',\Sigma')$, define 
reindexing (by precomposition) maps 
$$
\kappa^\M_{\Code} : \M_{R'}({\Code})\to \M_R({\Code}),
\qquad
(\kappa^\M_{\Code}(m))(r) \defeq m(\kappa(r)),
$$
$$
\kappa^\Omega : \Omega_{R'}\to \Omega_R,
\qquad
(\kappa^\Omega(\alpha))(r) \defeq \alpha(\kappa(r)).
$$
Intuitively, these maps reinterpret an $R'$-indexed computation or observable over $R$ by feeding it the rewired tape $\kappa(r)$. Thus a computation or predicate defined over $R'$ can be viewed over $R$ via precomposition.
This reindexing operation transports both computations and truth values along the tape map $\kappa$. 
We extend $\kappa^\Omega$ pointwise to tape indexed predicates:
for $\varphi':\Code\to \Omega_{R'}$ define $\kappa^\ast\varphi':\Code\to\Omega_R$ by
$
(\kappa^\ast\varphi')(c)\defeq \kappa^\Omega(\varphi'(c))
$ 
and likewise for $\psi':\Code \to \Omega_{R'}$. 

We first show that reindexing along $\kappa$ commutes with the tape modality.
\begin{lemma}\label{lem:kappa-commutes-modality}
For all sets $X$, all $m\in\M_{R'}(X)$, and all $\psi':X\to\Omega_{R'}$,
$$
\kappa^\Omega\Big(\after{x}{m}{\psi'(x)}\Big)
\;=\;
\after{x}{\kappa^\M_{X}(m)}{(\kappa^\ast\psi')(x)}
$$
\end{lemma}

A tape map $\kappa:R\to R'$ induces a semantic reindexing operation by precomposition:
it rewires an $R'$-tape computation into an $R$-tape computation, and it rewires
$R'$-truth values into $R$-truth values.
However, an evidenced-frame entailment is not merely a semantic inequality: it is
\emph{tracked} by a code $e\in\Code$ that serves as uniform evidence across all inputs.
Thus, to transport entailments between tape spaces, it is not enough to reindex computations
and predicates semantically, one must also specify how an evidence code over $R'$ is translated
into an evidence code over $R$.
We therefore require a code translation compatible with semantic reindexing.

\begin{definition}
\label{def:realizable-tapemap}
A measurable $\kappa:(R,\Sigma)\to (R',\Sigma')$ is \emph{realizable} in $\Code$ if there exists a a function
$
\tr_\kappa:\Code\to\Code$ such that for all $e,c\in\Code$,
$
\kappa^\M_{\Code}(e\cdot_{R'} c)
\;=\;
\tr_\kappa(e)\cdot_R c
$.
\end{definition}
Thus realizability of $\kappa$ means that semantic reindexing along $\kappa$ admits a
uniform translation of evidence codes: for every evidence code $e$ over the $R'$-tape model,
the translated code $\tr_\kappa(e)$ over the $R$-tape model has exactly the reindexed behavior.
Under this assumption, entailments proved in the $R'$-tape evidenced frame can be
transported to entailments in the $R$-tape evidenced frame by reindexing predicates and
translating evidence.

\begin{theorem}
\label{thm:tapemap-entailment-transport}
Assume $\kappa:(R,\Sigma)\to (R',\Sigma')$ is realizable with witness $\tr_\kappa:\Code\to\Code$.
Then for all predicates $\varphi',\psi':\Code\to\Omega_{R'}$ and all $e\in\Code$,
$$
  \ef_{R'}\;\vdash\;\entails{\varphi'}{e}{\psi'}
  \quad\Longrightarrow\quad
  \ef_{R}\;\vdash\;\entails{\kappa^*\varphi'}{\tr_\kappa(e)}{\kappa^*\psi'}.
$$
\end{theorem}
\begin{proof}
Assume $\ef_{R'}\;\vdash\;\entails{\varphi'}{e}{\psi'}$.
Thus for all $c\in\Code$ we have
$
\varphi'(c)\ \le\ 
\after{a}{(e \cdot_{R'} c)}{\psi'(a)}
$ in $\Omega_{R'}$. 
Applying $\kappa^\Omega$ to both sides and using monotonicity of $\kappa^\Omega$ (pointwise order) obtains 
$
\kappa^\ast\varphi'(c)\ = \kappa^{\Omega}(\varphi'(c))\ \le\ \kappa^\Omega\Big(\after{a}{(e \cdot_{R'} c)}{\psi'(a)}\Big)
$.
Now by Lemma~\ref{lem:kappa-commutes-modality} with $m=(e\cdot_{R'} c)$ and $\psi'$, we obtain 
$
\kappa^\ast\varphi'(c)\ \le\ 
\after{a}{\kappa^\M_{\Code}(e \cdot_{R'} c)}{(\kappa^\ast\psi')(a)} 
$.
Finally, the realizability of $\kappa$ (Def.~\ref{def:realizable-tapemap}) allows rewriting 
$\kappa^\M_{\Code}(e\cdot_{R'} c)$ as $(\tr_\kappa( e)\cdot_R c)$, yielding
$
\kappa^\ast\varphi'(c)\ \le\ 
\after{a}{(\tr_\kappa( e)\cdot_R c)}{(\kappa^\ast\psi')(a)}
$, 
and thus  $\ef_{R}\;\vdash\;\entails{\kappa^*\varphi'}{\tr_\kappa( e)}{\kappa^*\psi'}$.
\end{proof}

In general, Theorem~\ref{thm:tapemap-entailment-transport} only requires a code translation
$\tr_\kappa:\Code\to\Code$ compatible with semantic reindexing.
In concrete language instances with an explicit tape interface, one often expects many computable
tape rewiring maps to be realizable in this sense. For stream tapes $R=2^{\Nat}$, typical examples
include bit flip, prefix dropping, even or odd splitting, and grouping into fixed size blocks.
More precisely, if the code language is closed under computable tape transformations, then such
maps admit corresponding translations $\tr_\kappa$ satisfying
Definition~\ref{def:realizable-tapemap}.
In stronger language instances, for example with a universal evaluator, one may
further ask that this translation be internally implemented by a code
$\mathsf{tmap}_\kappa\in\Code$. In that case, tape transport becomes an internal compiler pass:
given a code $e$, the program $\mathsf{tmap}_\kappa$ computes a rewired code $\tr_\kappa(e)$
whose behavior on $R$ tapes simulates that of $e$ on $R'$ tapes.

Although tape-level entailment is measure-agnostic, the extracted numeric reading depends on the chosen measure.
Thus, a tape map interacts with extraction by transporting measures.

\begin{definition}
    Given a measurable function $\kappa:(R,\Sigma)\to (R',\Sigma')$ and $\rho$,  a probability measure on $(R, \Sigma)$, i.e $\rho : \Sigma \to [0,1]$, define the pushforward measure on $(R', \Sigma')$ denoted as $\kappa_{\ast}\rho$,  by: 
    $
    \forall B \in \Sigma'. \kappa_{\ast}\rho(B) = \rho(\kappa^{-1}(B))
    $.
\end{definition}

\begin{lemma}[Extraction commutes with tape reindexing]\label{lem:extract-tape-pullback}
Let $\kappa:(R,\Sigma)\to(R',\Sigma')$ be measurable,  $\rho$ be a probability measure on $(R,\Sigma)$, and  $\rho' \defeq \kappa_{\ast}\rho$ be the pushforward measure on $(R',\Sigma')$.
Then for every measurable $\alpha':R'\to[0,1]$ we have
$
E_\rho(\kappa_\Omega(\alpha')) \;=\; E_{\rho'}(\alpha')
$.
\end{lemma}
Combining Lemma~\ref{lem:extract-tape-pullback} with the extraction soundness, the same evidence compilation along $\kappa$ that transports tape-level entailments also transports the corresponding extracted numeric bounds, provided we push the measure forward along $\kappa$.

We next specialize the tape space to infinite bit streams.
This fixes a canonical source of unbounded randomness and  provides a canonical measure-preserving splitting operation.
\section{Stream tapes and unbounded randomness}
\label{sec:tape:streams}

So far the tape evidenced frame has been developed for an arbitrary measurable tape space.
We now instantiate the tape space by infinite bit streams.
This serves two purposes: it supports unbounded consumption of randomness, and it provides a canonical, measure-preserving splitting map.
The splitting map is available semantically in the stream model, and, when it is realizable in the sense of Definition~\ref{def:realizable-tapemap}, it yields an independence discipline at the evidenced-frame level.
We will later combine this transport principle with the fact that splitting is measure-preserving in order to move between tape-level reasoning about independent draws and standard probabilistic statements.

Let
$
R \;=\; \{0,1\}^{\mathbb N}
$. 
For a finite bit word $w\in\{0,1\}^{<\mathbb N}$ write $[w]\subseteq R$ for the cylinder set of streams with prefix $w$.
Let $\Sigma$ be the $\sigma$-algebra generated by all cylinder sets, and let $\mu$ be the fair coin product measure,
uniquely determined by
$
\mu([w]) \;=\; 2^{-|w|}
$. 
This section takes the tape evidenced frame $\ef_\mathsf{tape}$ specialized to $(R,\Sigma,\mu)$.

\subsection{Why unbounded randomness matters: von Neumann unbiasing}

Many basic probabilistic procedures do not have a fixed a priori bound on how many random draws they will need.
A standard example is von Neumann's procedure for extracting an unbiased coin from a biased coin:
read the tape in pairs until seeing $01$ or $10$~\cite{VN}.
(Equivalently, repeat until the two bits differ.)
The point of this example is not primarily the classical fact that von Neumann's procedure produces a fair coin,
but to illustrate a proof pattern that is specific to our evidenced frame:
fairness is established by a  transport of evidence (a realizable tape map), and only then
turned into an ordinary numeric probability statement by extraction.

Let the observable outcomes be
$
\mathsf{Val} \defeq \{\mathsf{H},\mathsf{T}\}$. 
Define an extensional tape computation
$
\mathsf{vn}:R\to\mathsf{Val}_\bot
$
by scanning the tape in consecutive pairs $(r(2k),r(2k+1))_{k\in\mathbb N}$:
$\mathsf{vn}(r)=\mathsf{H}$ if the first pair $(r(2k),r(2k+1))$ with unequal bits is $01$,
$\mathsf{vn}(r)=\mathsf{T}$ if that first such pair is $10$,
and $\mathsf{vn}(r)=\bot$ if no such pair exists.
Because $\mathsf{Val}_\bot$ is discrete, $\mathsf{vn}$ is measurable iff each fiber is measurable. 
Indeed,
$\mathsf{vn}^{-1}(\{\mathsf{H}\})$ is the disjoint union over $k\in\mathbb N$ of the cylinders describing
$k$ initial nondecisive pairs in $\{00,11\}$ followed by the decisive pair $01$, and similarly for $\mathsf{T}$ and $\bot$.

Let $\mathsf{flip}:R\to R$ be the bit flip map
$
\mathsf{flip}(r)(n)\defeq 1-r(n)
$. 
Intuitively, $\mathsf{flip}$ exchanges the labels $\mathsf{H}$ and $\mathsf{T}$: flipping all bits swaps the decisive patterns
$01$ and $10$ while preserving the nondecisive patterns $00$ and $11$.
Formally, define $\mathsf{swap}:\mathsf{Val}_\bot\to\mathsf{Val}_\bot$ by
$$
\mathsf{swap}(\mathsf{H})\defeq\mathsf{T},
\qquad
\mathsf{swap}(\mathsf{T})\defeq\mathsf{H},
\qquad
\mathsf{swap}(\bot)\defeq\bot.
$$
Then for all $r\in R$,
$
\mathsf{vn}(\mathsf{flip}(r)) \;=\; \mathsf{swap}(\mathsf{vn}(r))
$.

Outcomes are observed using crisp postconditions.
Let $P_{\mathsf{H}}\defeq\{\mathsf{H}\}$ and $P_{\mathsf{T}}\defeq\{\mathsf{T}\}$, and let
$\mathbf{1}_{P_{\mathsf{H}}},\mathbf{1}_{P_{\mathsf{T}}}:\mathsf{Val}\to[0,1]$ be their indicators.
Write $(\cdot)^\sharp$ for the tape-invariant lift to predicates $\mathsf{Val}\to(R\to[0,1])$.
Consider the tape truth values
$$
\alpha_{\mathsf{H}}(r)\defeq
\Big(
\after{a}{\mathsf{vn}}{(\mathbf{1}_{P_{\mathsf{H}}})^\sharp(a)} 
\Big)(r),
\qquad
\alpha_{\mathsf{T}}(r)\defeq
\Big(
\after{a}{\mathsf{vn}}{(\mathbf{1}_{P_{\mathsf{T}}})^\sharp(a)} 
\Big)(r),
$$
which are simply the crisp tests for whether $\mathsf{vn}$ returns $\mathsf{H}$ or $\mathsf{T}$ on tape $r$
(with divergence contributing $0$ by definition of the modality).
The symmetry $\mathsf{vn}(\mathsf{flip}(r))=\mathsf{swap}(\mathsf{vn}(r))$ implies the pointwise identity
$
\alpha_{\mathsf{H}}(\mathsf{flip}(r)) \;=\; \alpha_{\mathsf{T}}(r)$ 
for all $r\in R$.
Equivalently,
$
\mathsf{flip}^\Omega(\alpha_{\mathsf{H}}) \;=\; \alpha_{\mathsf{T}}$.

Now assume $\mathsf{flip}$ is realizable, witnessed by the translation
$\tr_\mathsf{flip} : \Code \to \Code$.
Then Theorem~\ref{thm:tapemap-entailment-transport} transports any evidenced entailment along this tape symmetry:
whenever $e$ witnesses an entailment over the $\mathsf{flip}$-rewired specification, the translated code $\tr_\mathsf{flip}(e)$ witnesses the corresponding entailment over the original tape model.
Thus the tape symmetry is reflected not only at the level of truth values, but also at the level of uniform evidence.

Since $\mathsf{flip}$ is measure preserving for the fair product measure $\mu$, extraction commutes with reindexing:
for measurable $\beta:R\to[0,1]$,
$
E_\mu(\beta\circ\mathsf{flip})\;=\;E_\mu(\beta)
$.
Applying this to $\beta=\alpha_{\mathsf{H}}$ and using $\alpha_{\mathsf{H}}\circ\mathsf{flip}=\alpha_{\mathsf{T}}$ gives
$
E_\mu(\alpha_{\mathsf{H}})\;=\;E_\mu(\alpha_{\mathsf{T}})
$.
If moreover $\mathsf{vn}$ terminates $\mu$-almost surely (a standard fact), then $E_\mu(\alpha_{\mathsf{H}})+E_\mu(\alpha_{\mathsf{T}})=1$,
hence each equals $1/2$.
In other words, the classical fairness conclusion is obtained by extracting a numeric statement from an evidenced frame level symmetry, 
 while realizability of $\mathsf{flip}$ ensures that the corresponding evidence transports uniformly across the rewiring.

\subsection{A splitting discipline for independence}

The tape monad threads a single tape through sequential composition.
This makes correlation explicit, but it also means that independence is not automatic:
two subcomputations executed in sequence will, by default, read from the same tape.
To model independent draws, it is convenient to transport the resulting entailments back
to one tape along a realizable splitting map, see, e.g.,~\cite{DeterminisitcStreamSample2023}.

Fix the stream tape measure $\mu$.
Any tape computation $m:R\to X_\bot$ induces an extensional \emph{law} on outcomes by pushforward:
for measurable $A\subseteq X_\bot$ define
$
  \mathsf{law}(m)(A) \;\defeq\; \mu\big(m^{-1}(A)\big)
$. 
When $X_\bot$ is discrete, we may equivalently write $\mathsf{law}(m)(z)=\Pr_{r\sim\mu}[m(r)=z]$.
Intuitively, $\mathsf{law}(m)$ is what an extensional probabilistic semantics remembers about $m$:
it forgets \emph{which} tapes lead to which outcomes and records only their probabilities.

Sequential composition in the tape monad threads the \emph{same} tape through the continuation, so correlation is explicit.
As a result, taking laws does not in general commute with sequential composition:
running $m$ and then $f$ against the same tape is not the same as running them with \emph{fresh independent} randomness at each step.
A standard way to enforce independent sequencing \emph{inside} a tape semantics is therefore to equip tapes with a splitting map,
and to use it as an explicit independence discipline.

Define the even/odd splitting map
$
\mathsf{split}:R\to R\times R
$
by
$$
\mathsf{split}(r) \;=\; (r_{\mathsf{even}},r_{\mathsf{odd}}),
\qquad
r_{\mathsf{even}}(n)\defeq r(2n),
\qquad
r_{\mathsf{odd}}(n)\defeq r(2n+1).
$$

\begin{lemma}
\label{lem:split-product}
The pushforward of $\mu$ along $\mathsf{split}$ is the product measure $\mu\times\mu$ on $R\times R$.
In particular, $r_{\mathsf{even}}$ and $r_{\mathsf{odd}}$ are independent and each has distribution $\mu$.
\end{lemma}
Using $\mathsf{split}$ we can define a derived sequencing operation on tape computations that enforces fresh randomness.
Given $m:R\to X_\bot$ and $f:X\to (R\to Y_\bot)$, define
$$
  (m \kleisli_{\mathsf{split}} f)(r) \;\defeq\;
  \text{let } (r_1,r_2)=\mathsf{split}(r)\ \text{in}\ 
  \begin{cases}
    \bot & \text{if } m(r_1)=\bot,\\
    f(x)(r_2) & \text{if } m(r_1)=x.
  \end{cases}
$$

\begin{proposition}
\label{prop:law-split-seq}
For all measurable $B\subseteq Y_\bot$,
$$
  \mathsf{law}(m \kleisli_{\mathsf{split}} f)(B)
  \;=\;
  \int_{X_\bot} \mathsf{law}(f(x))(B)\, d\mathsf{law}(m)(x).
$$
In particular, split sequencing realizes the intended ``fresh independent randomness'' interpretation of sequential composition.
\end{proposition}
\begin{proof}
Expand the LHS as $\mu(\{r \mid (m\kleisli_{\mathsf{split}} f)(r)\in B\})$ and rewrite it using the definition of
$\kleisli_{\mathsf{split}}$.
Apply Lemma~\ref{lem:split-product} to replace sampling $r\sim\mu$ with sampling $(r_1,r_2)\sim \mu\times\mu$,
and then use the law of total probability (equivalently, conditional expectation) to separate the $m$-part (on $r_1$) from the $f$-part (on $r_2$).
\end{proof}

To use split as an independence discipline inside the evidenced frame, we apply the general tape map transport theorem
of Section~\ref{sec:tape:tmap}.
Let $R^{(2)}\defeq R\times R$ with product $\sigma$-algebra and product measure $\mu\times\mu$.
The two tape evidenced frame uses the same code set $\Code$, but application now returns computations over $R^{(2)}$:
for $e,c\in\Code$,
$
(e\cdot_{R^{(2)}} c) : R^{(2)} \to \Code_\bot
$. 
Precomposition along $\mathsf{split}:R\to R^{(2)}$ then rewires such a two tape computation into a one tape computation by feeding it
the pair $(r_{\mathsf{even}},r_{\mathsf{odd}})$.

The measurable map $\mathsf{split}$ is canonical at the semantic level, but its realizability is
an additional assumption on the code language rather than a consequence of the abstract tape
semantics.
Assume therefore that $\mathsf{split}$ is realizable, witnessed by a   translation 
$\tr_{\mathsf{split}}:\Code\to \Code$.
Then Theorem~\ref{thm:tapemap-entailment-transport} yields, for all predicates
$\varphi',\psi':A\to\Omega_{R^{(2)}}$ and all $e\in A$,
$$\ef_{R^{(2)}}\vdash \entails{\varphi'}{e}{\psi'}
\quad\Longrightarrow\quad
\ef_R\vdash \entails{\mathsf{split}^*\varphi'}{\tr_{\mathsf{split}}(e)}{\mathsf{split}^*\psi'}.$$
Thus, an entailment proved via an interface that assumes access to two independent tapes can be transported into an
entailment for a single tape by 
reindexing predicates along
$\mathsf{split}$ and translating evidence accordingly.
Lemma~\ref{lem:split-product} ensures that, after  probabilities extraction, this transport matches the intended
semantics of independent randomness under the product measure.

\begin{example}[Error reduction by majority]
This example shows how our splitting discipline supports the standard ``repeat and take majority'' amplification argument.
Since tape-based repetition reuses the same tape, independence is not automatic, so we build the majority-of-$k$ verifier on $k$ split sub-tapes (to enforce independence), transport the construction back to a single-tape program using realizability of splitting, and read off the usual probability improvement by expectation extraction.

Fix a tape probability space $(R,\Sigma,\rho)$ and the induced evidenced frame $\ef_R$.
Let $P\subseteq \Code_R$ be a crisp set of accepting outputs, and write $\psi_P(c)(r)\defeq \mathbf{1}_P$.
Let $e\in\Code_R$ be a (tape-dependent) randomized verifier, i.e., a code that may consult the tape when run.

For $k\ge 1$ let $R^{(k)}\defeq R^k$ with the product measure $\rho^{\otimes k}$ and the induced evidenced frame $\ef_{R^{(k)}}$.
Let $e_{k,t}\in \Code_{R^{(k)}}$ be the code that runs the input code $k$ times (once per tape component) and accepts if at least $t$ of the runs accept.
Define the threshold predicate
$$
\varphi^{(k)}_{\ge t}(c)(r_1,\dots,r_k)
\;\defeq\;
\mathbf{1}_{\big[\#\{\,i \mid (e\cdot c)(r_i)\in P\,\}\ge t\big]}.
$$
By the threshold construction, we have an evidenced entailment in $\ef_{R^{(k)}}$:
$$
\ef_{R^{(k)}}\vdash
\entails{\varphi^{(k)}_{\ge t}}{e_{k,t}}{\psi_P},
$$

Now let $\kappa\defeq \mathsf{split}_k:R\to R^k$ be the $k$-ary splitting map,
$
\mathsf{split}_k(r)\defeq$ $(r_0, ..., r_{k-1})
$, where $r_i(n) = r(k \cdot n + i)$.
Assume $\mathsf{split}_k$ is realizable, with associated evidence translation $
\tr_k:\Code\to \Code$. 
Then Theorem~\ref{thm:tapemap-entailment-transport} yields the single-tape entailment
$$\ef_R\vdash
\entails{\mathsf{split}_k^*\varphi^{(k)}_{\ge t}}{\tr_k(e_{k,t})}{\mathsf{split}_k^*\psi_P}.$$

The logical entailment is tape-by-tape, and we induce it with extracted probabilistic reading.
Soundness of extraction gives
$
E_\rho\!\bigl(\mathsf{split}_k^*\varphi^{(k)}_{\ge t}(c)\bigr)
\;\le\;
E_\rho\!\Bigl(
\bigl\langle \Diamond a \leftarrow (\tr_k(e_{k,t})\cdot c)\bigr\rangle
\psi_P(a)
\Bigr)
$ for any $c\in\Code_R$. 
Unfolding $\mathsf{split}_k^{*}$ and the fact that $\mathsf{split}_k$ pushes $\rho$ forward to $\rho^{\otimes k}$ obtains
$
E_\rho\big(\mathsf{split}_k^{*}\varphi^{(k)}_{\ge t}(c)\big)
=
E_{\rho^{\otimes k}}\big(\varphi^{(k)}_{\ge t}(c)\big),
$ 
  hence
$
E_\rho\!\Big(
\after{a}{(\tr_k(e_{k,t})\cdot c)}{\psi_P(a)}
\Big)
\ \ge\
E_{\rho^{\otimes k}}\big(\varphi^{(k)}_{\ge t}(c)\big)
$.

As a concrete instance, consider the majority of three by taking  $k=3$ and $t=2$.
Let
$
p_c \;\defeq\; E_\rho\!\Big(
\after{a}{(e \cdot c)}{\psi_P(a)}
\Big)
$ 
be the one-shot acceptance probability.
Then
$
E_{\rho^{\otimes 3}}\big(\varphi^{(3)}_{\ge 2}(c)\big)
=
3p_c^2-2p_c^3
$ is 
the usual majority amplification, which strictly improves error when $p_c>\tfrac12$.
Thus, by combining splitting (to get independent runs), realizability of the corresponding tape map, and
expectation extraction, we obtain an evidence-tracked derivation of error reduction.
\end{example}

\section{Finitary distributions as an abstraction of tapes}
\label{sec:dist-final}

The tape-based evidenced frame exposes correlations by reusing the same tape and supports evidence compilation along realizable tape maps.
It is nevertheless useful to record a coarser view that forgets the tape and keeps only the \emph{law} (output distribution) of a computation.
This section packages that view as a finitary distribution semantics and relates it back to the tape evidenced frame.
Because the tape model is naturally quotiented by almost-sure equality, the fragment that survives most cleanly without expectation is the probability-one fragment, so we work with Boolean truth values and the corresponding `must' modality.

\subsection{A finitary distribution monadic core}

For countable $X$, let $\mathcal{D}_{\mathrm{fin}}(X\bot)$ be the set of finitely supported distributions on $X_\bot$:
$$
\mathcal{D}_{\mathrm{fin}}(X_\bot)
\;=\;
\Big\{\,\mu:X_\bot\to[0,1]\ \Big|\ \mathrm{supp}(\mu)\ \text{finite and}\ \sum_{z\in X_\bot}\mu(z)=1\,\Big\}.
$$
The unit is the Dirac distribution, and Kleisli bind is mixture:
$
(\mu \kleisli f)(w) \defeq \sum_{z\in X_\bot} \mu(z)\cdot f(z)(w)
$. 
We take $\Omega_{\mathrm{bool}}\defeq\{0,1\}$ with the usual order, hence a complete Heyting algebra.
Predicates $\varphi:X_\bot\to\Omega_{\mathrm{bool}}$ are characteristic functions of subsets of $X_\bot$.

\begin{definition}
\label{def:must-modality}
Define $\diamondop_{must}$ by
$$
\aftermust{x}{\mu}{\varphi(x)}=1
\quad\Longleftrightarrow\quad
\forall z\in X_\bot.\ \mu(z)>0\ \Rightarrow\ \varphi(z)=1.
$$
\end{definition}
Equivalently, $\aftermust{x}{\mu}{\varphi(x)}=1$ iff $\mu(\{z\mid \varphi(z)=1\})=1$.
Thus must expresses support-based safety: a postcondition holds iff it holds for all outcomes that may occur with nonzero probability.
In particular, excluding $\bot$ expresses almost-sure termination.

\begin{lemma}
\label{lem:must-modality}
The operator $\diamondop_{must}$ is an $M$-modality.
\end{lemma}

Note that we can replace $\Omega_{\mathrm{bool}}$ by the complete Heyting chain $\Omega=[0,1]$ (with G\"odel implication) and keep $\mathcal{D}_{\mathrm{fin}}$.
The induced must modality is then the support-infimum transformer
$
\aftermust{x}{\mu}{\varphi(x)}
\;\defeq\;
\inf\{\,\varphi(z)\mid \mu(z)>0\,\}$.
This is an $M$-modality, but it is best read as demonic or safety-style quantitative refinement (a worst-case score across possible outcomes), not as a probability bound.
Since our tape model focuses on almost-sure extensionality, we use the Boolean must fragment as the canonical distribution-level abstraction and leave other quantitative choices to future work.

\begin{definition}
Let $\ef_{\mathrm{dist}}$ denote the evidenced frame induced by the monadic core $\MC \;=\; (\Code,\Omega_{\mathrm{bool}},\diamondop_{must})$ over the monad $\mathcal{D}_{\mathrm{fin}}$.
\end{definition}
$\ef_{\mathrm{dist}}$ provides a compact law-level evidenced frame whose propositions range over deterministic codes and whose entailments express support-based (probability-one) properties.
Thus, the proposition fiber is $\Phi_{\mathrm{dist}}\defeq\Code\to\Omega_{\mathrm{bool}}$, so predicates range over deterministic codes.
Entailment has the usual uniform refinement form:
$$
  \entails{\varphi}{e}{\psi}
  \quad\Longleftrightarrow\quad
  \forall c\in\Code.\ 
  \varphi(c)=1
  \Rightarrow
  \aftermust{a}{(e \cdot_0 c)}{\psi(a)}=1.
$$
Thus predicates express support-based properties such as safety, invariants over possible outputs, and probability-one termination (by excluding $\bot$).

Note that one can also take codes themselves to be finitary distributions.
Let $\Code^* \defeq \mathcal{D}_{\mathrm{fin}}(\Code)$ and suppose we have an MCA structure with application
$$
  \cdot_*:\Code^*\times \Code^* \to \mathcal{D}_{\mathrm{fin}}(\Code^*_\bot)
  \;=\;
  \mathcal{D}_{\mathrm{fin}}\big(\mathcal{D}_{\mathrm{fin}}(\Code)_\bot\big).
$$
Equipping this MCA with the same $\Omega_{\mathrm{bool}}$ and $\diamondop_{must}$ yields an alternative evidenced frame
whose propositions are predicates $\Code\to\Omega_{\mathrm{bool}}$ on internal laws.
Although truth values remain Boolean, such predicates can express numerical constraints on the internal distribution carried by a code.
For $E\subseteq\Code$ and threshold $p\in[0,1]$, define
$
  P_{E,p}(\nu)=1
  \quad\Longleftrightarrow\quad
  \sum_{x\in E}\nu(x)\ge p
$.

\subsection{From tapes to finitary laws}
We now relate the finitary distribution view to the tape model.
Fix a tape probability space $(R,\Sigma,\rho)$.
For any measurable tape computation $m:R\to X_\bot$, define its law, i.e., its pushforward output distribution under $\rho$,
$\mathsf{push}(m):X_\bot\to[0,1]$ by
$$
\mathsf{push}(m)(z) \defeq \rho(\{\,r\in R \mid m(r)=z\,\}).
$$
This forgets which tapes lead to which outcomes, keeping only their probabilities.
If $\mathsf{push}(m)$ has finite support then $\mathsf{push}(m)\in\mathcal{D}_{\mathrm{fin}}(X_\bot)$.

If $m,n:R\to X_\bot$ are equal $\rho$-almost surely, then $\mathsf{push}(m)=\mathsf{push}(n)$.
Equality in law is strictly coarser than almost-sure equality: it identifies computations with the same output distribution,
even when they differ on how they consume and reuse randomness.

For finitary laws, probability-one validity coincides with a must-style, support-based reading.
Given $P\subseteq X_\bot$ and $\mu\in\mathcal{D}_{\mathrm{fin}}(X_\bot)$, write
$$
\mu \models_{\must} P \quad:\!\!\Longleftrightarrow\quad \forall z\in X_\bot.\ \mu(z)>0 \Rightarrow z\in P.
$$
Equivalently, $\mu \models_{\must} P$ iff $\sum_{z\in P}\mu(z)=1$.

This gives a clean interface from tape-based reasoning to the finitary law view:
whenever a tape-level argument establishes that a crisp postcondition $P$ holds with probability $1$
(for instance via the almost-sure quotient and extraction), the induced finitary law $\mathsf{push}(m)$ satisfies
the must judgment $\mathsf{push}(m)\models_{\must} P$.
Conversely, must judgments about $\mathsf{push}(m)$ are stable under the almost-sure quotient.

\section{Conclusion}\label{sec:conc}

We developed an evidence-tracked, tape-based semantics for probabilistic computation, where randomized programs consume an explicit random tape and sequencing reuses it, making correlation explicit.
Entailments are uniform tape-by-tape refinements witnessed by evidence transformers.
After fixing a tape measure, probabilistic content is recovered by extraction and soundness, and we quotient truth values by almost-sure equality to obtain a probability-one fragment.
We also identify two intensional principles lost at the level of laws: proof-relevant transport along realizable tape rewiring maps and a canonical splitting discipline for stream tapes.
Finally, we relate tape reasoning to a probability-one must abstraction via pushforward.

Several lines of work make randomness explicit through streams, tapes, or random variables, including deterministic stream sampling~\cite{deterministicStream2023} and tape-based denotational accounts of probabilistic programming~\cite{Commutative2017Staton}. 
We share this randomness-as-input perspective, which is particularly useful for higher-order probabilistic computation where measure-theoretic models face well-known structural difficulties, such as the failure of the category $\Meas$ to be cartesian closed and the resulting use of quasi-Borel spaces and measurable cones~\cite{qbs,MeasurableCones2017}. 
However, our goal and interface are different: we use realizability and evidenced frames to obtain proof-relevant entailments, so reasoning steps carry evidence and support transport along realizable tape maps and splitting principles for independence. 
Extensional probabilistic claims are extracted only later, after fixing a tape measure and passing to laws. 
Compared with probabilistic Hoare, expectation-based, relational, and outcome-oriented logics~\cite{Barthe2013,McIverMorgan05,OutcomeLogic,zilberstein2026probabilistic}, our approach is more intensional: outcome-style reasoning appears as a Boolean fragment, while richer truth values support quantitative evidence-tracked reasoning.

Several extensions are natural.
On the semantic side, we currently restrict to countable outcome types so that measurability of Kleisli composition follows from countable unions. Extending the tape monad to categories of measurable spaces such as $\Meas$ or standard Borel spaces~\cite{Aumann1961,qbs,MeasurableCones2017} would allow richer outputs while preserving extraction and almost sure reasoning.
Our transport theorems isolate realizability of a tape map as the key condition ensuring that entailments and evidence compile along tape rewiring, and it would be valuable to characterize this condition in concrete language instances.
At the extensional level, we singled out the probability one must fragment as the abstraction commuting with the almost sure quotient. A natural next step is to identify the largest law-level fragment soundly obtainable from tape proofs, possibly via extraction, and to understand when such law-level principles can be reflected back into tape refinements under additional structure.
Finally, although the present paper is semantic in focus, the induced tripos also points toward a proof relevant higher-order proof theory: the non-tape fragment should follow existing evidence tracked systems for evidenced frames, e.g.~\cite{effhol}, while the tape-specific content would appear as additional admissible principles for transport and splitting.
Developing such a syntactic account for a concrete tape access language, together with soundness and completeness results, is left for future work.

\clearpage

\bibliographystyle{plainurl}
\bibliography{refs}

\clearpage
\appendix
\section{Appendix}

This appendix contains proofs for the paper.

\begin{proof}[Proof of Lemma~\ref{lem:reader-monad}]
$\unit_X(x)$ is constant, hence measurable.
If $m$ and all $f(x)$ are measurable, then $(m\kleisli f)^{-1}(\{y\})$ is a union of sets of the form
$m^{-1}(\{x\})\cap f(x)^{-1}(\{y\})$, hence measurable.
The monad laws are the usual reader-monad equalities.
\end{proof}

\begin{proof}[Proof  of Lemma~\ref{lem:diamop-modality}]
All three axioms hold pointwise in $r\in R$.
After-Return holds because $\unit_X(x)(r)=x$.
After-Bind holds because both sides case split on whether $m(r)=\bot$ or $m(r)=x$.
Internal Monotonicity holds because the Heyting structure on $\OmegaTV$ is pointwise:
it reduces to the corresponding inequality in $\Omega_0=[0,1]$ for each fixed tape $r$.
\end{proof}

\begin{proof}[Proof of Lemma~\ref{lem:diamond-meas}]
For $q\in\mathbb Q$,
\[
\{r\mid \after{x}{m}{\psi(x)}(r)>q\}
=\bigcup_{x\in X}\Big(m^{-1}(\{x\})\cap \{r\mid \psi(x)(r)>q\}\Big),
\]
a countable union of measurable sets.
\end{proof}

\begin{proof}[Proof of Lemma~\ref{lem:modality-respects-as}]
For each $x\in X$ let $N_x\defeq \{\,r\in R \mid \psi(x)(r)\neq \psi'(x)(r)\,\}$, so $\rho(N_x)=0$.
Since $X$ is countable, $N\defeq \bigcup_{x\in X}N_x$ is measurable and $\rho(N)=0$.
Fix $r\notin N$. If $m(r)=\bot$, then both sides evaluate to $0$ at $r$.
If $m(r)=x$, then $\psi(x)(r)=\psi'(x)(r)$ by choice of $N$, hence the two sides agree at $r$.
Thus the two resulting truth values agree outside a $\rho$-null set, i.e.\ are $\equiv_{\mathrm{a.s.}}$.    
\end{proof}

\begin{proof}[Proof of Lemma~\ref{prop:quotient-modality-axioms}]
Choose measurable representatives throughout, and use that the corresponding laws hold \emph{pointwise in $r$} for $\diamond$ on $\Omega_{\meas}$.
Monotonicity follows by collecting the countably many null sets witnessing $\psi(x)\le\chi(x)$ and arguing as in Lemma~\ref{lem:modality-respects-as}.
Unit and bind are equalities of measurable functions $R\to[0,1]$ holding for every $r$, hence their classes are equal in $\Omega_{\mathrm{a.s.}}$.
\end{proof}

\begin{proof}[Proof of Lemma~\ref{lem:kappa-commutes-modality}]
Both sides are equal as functions $R\to[0,1]$.
If $m(\kappa(r))=\bot$ then both sides are $0$ and if $m(\kappa(r))=x\in X$ then both sides are $\psi'(x)(\kappa(r))$.
\end{proof}

\begin{proof}[Proof of Lemma~\ref{lem:extract-tape-pullback}]
Unfolding LHS results in
$$
    E_\rho(\kappa_\Omega(\alpha')) = \int_{R} \kappa_\Omega(\alpha'(r) d\rho(r) = \int_{R} \alpha'(\kappa(r)) d\rho(r)
$$
By definition of $\rho' = \kappa_{\ast}\rho$ (i.e $\rho'(B) = \rho(\kappa^{-1}(B)$ for all $B \in \Sigma'$) and by change-of-variables formula for pushforward measures (equivalently, the defining property of $\kappa_*\rho$), the above equals to
$$
\int_{R'} \alpha'(r')\,d\rho'(r') = E_{\rho'}(\alpha')
$$
Hence $ E_\rho(\kappa_\Omega(\alpha')) = E_{\rho'}(\alpha')$. as required.
\end{proof}

\begin{proof}[Proof of Lemma~\ref{lem:split-product}]
It suffices to check equality on cylinder rectangles $[u]\times[v]$.
The preimage $\mathsf{split}^{-1}([u]\times[v])$ constrains $|u|$ even positions and $|v|$ odd positions of $r$,
so it is a cylinder of length $|u|+|v|$ and has measure $2^{-(|u|+|v|)}$, which equals $\mu([u])\mu([v])$.
\end{proof}

\begin{proof}[Proof of Lemma~\ref{lem:must-modality}]
After-Return: $\delta_x$ has support $\{x\}$, so the must condition reduces to $\varphi(x)=1$.
After-Bind: an outcome $w$ has positive probability under $\mu\kleisli f$ iff there exists $z$ with $\mu(z)>0$
and $f(z)(w)>0$, so the must condition for $\mu\kleisli f$ is equivalent to:
for all $z$ in the support of $\mu$, the must condition holds for $f(z)$.
Internal Monotonicity is immediate for Boolean truth values.
\end{proof}

\end{document}